\documentclass[a4paper, 12pt]{article}
\usepackage[utf8]{inputenc}
\usepackage[initials, alphabetic]{amsrefs}
\usepackage{amsmath}
\usepackage{amsthm}
\usepackage[psamsfonts]{amssymb}
\usepackage{amsfonts}
\usepackage{color}
\usepackage{mathrsfs}
\usepackage{amssymb}
\usepackage{graphicx}
\usepackage{fancybox}
\usepackage{enumerate}
\usepackage{verbatim}
\usepackage{subfigure}
\usepackage{bm}
\usepackage{braket}
\usepackage{mdframed}
\usepackage{comment}

\providecommand{\keywords}[1]
{
  \small	
  \textbf{\textrm{Keywords}:} #1
}

\providecommand{\MSC}[1]
{
  \small	
  \textbf{\textrm{2020 Mathematics Subject Classification}:} #1
}

\theoremstyle{definition}
\newtheorem{dfn}{Definition}[section]
\newtheorem{rem}[dfn]{Remark}

\newtheorem{ass}[dfn]{Assumption}

\theoremstyle{plain}

\newtheorem{lem}[dfn]{Lemma}
\newtheorem{thm}[dfn]{Theorem}

\theoremstyle{remark}

\theoremstyle{remark}

\begin{document}
\title{The Thermodynamic Approach to Whole-Life Insurance:
A Method for Evaluation of Surrender Risk}
\author{Jir\^o Akahori, Yuuki Ida, Maho Nishida, and Shuji Tamada \\
Department of Mathematical Sciences, Ritsumeikan University}
\date{December\, 2020}
\maketitle{}
\abstract{We introduce a collective model 
for life insurance 
where 
the heterogeneity of 
each insured, including 
the health state, 
is modeled by a 
diffusion process. This model is influenced by 
concepts in statistical mechanics.
Using the proposed framework, 
one can describe the total pay-off as 
a functional of the diffusion process, 
which can be used to derive a 
level premium that evaluates 
the risk of lapses 
due to
the so-called adverse selection.
Two numerically tractable models 
are presented to 
exemplify the flexibility of the proposed framework.}

\ 


\MSC{91G05, 60J70}

\section{Introduction}

The risk of lapses, 
also referred to as surrender risk, 
is the instability
associated with unexpected lapses 
of insurance contracts, 
which may result in 
a huge loss for the insurer.  
As insurance contracts are collective in nature, 
we can divide the cause
of lapses into two categories: 
homogeneous and heterogeneous (among the policy holders). 
The former is basically due to macro-economic shocks
--- typically changes in interest rates, recessions, inflation, and so forth. 
The surrender risk arising from 
interest rate fluctuations, for example, can be evaluated using option-pricing type 
technologies (see e.g. \cite{Alb-Geman}). 

The latter, i.e. heterogeneous causes, can be further divided into two groups:  
economic and non-economic. 
A policyholder may surrender
because she is unexpectedly short of 
money while the economy is, as a whole, healthy. 
Such a cause is classified as heterogeneous-economic (see e.g. \cite{MLM2011}). 

Among non-economic causes, 
demographic heterogeneity, 
which is variation in the force of mortality, has been a central issue 
both academically and practically, 
as it has been recognized as a (potential) cause of the so-called moral hazard or adverse selection.
The risk from heterogeneity 
has been recognized in insurance
(since the 19-th century!), 
as is pointed out in the seminal paper by G.A. Akerlof
\cite{Akerlof}, 
which made clear the role of the ``asymmetry of information" in adverse selection.
M. Rothschild and Nobel laureate J. Stiglitz proved 
the non-existence of an 
economic equilibrium 
under asymmetry 
in \cite{Rot-Stig}. According to that paper, in a ``rational world", 
adverse selection leads 
to the non-existence of an equilibrium,
implying potential instability of a life insurance contract. 

Since then there have been many studies on life insurance, 
however most empirical studies 
do not seem to observe the predicted negative effects of adverse selection 
(see e.g. \cite{CP}). 
Some recent studies, like \cite{Meza-Webb},
suggest the effect of so-called
advantageous selection:
healthier people might be more risk averse. 
Thus many theoretical studies 
have been rather interested in 
modelling 
the effect of heterogeneity on the surrender risk without 
the hypothesis of rationality. 
One of the main streams can be found 
in direct modelling of 
the dependence of demographic 
heterogeneity and surrender inclination. 
In most cases, however, 
this is modeled by 
a simple stochastic model,
a discrete-time, discrete-state
Markov chain at best; see e.g. 
\cite{Bl}, \cite{Jones1998}, 
\cite{term}, 
and more recently \cite{Adams},
to name a few.  
By contrast, 
the present paper 
proposes a model using diffusion processes. 

We will start in section \ref{secsibou1}
with introducing the basic framework of our model, without mortality or surrender risk, 
where the heterogeneity of each insured person is 
modeled by a multi-dimensional diffusion process aiming to describe the 
various causes discussed above. 
The introduced diffusion processes 
naturally define a probability measure, which describes the distribution 
of the ``heterogeneity", 
and we assume that 
it is approximated by
its infinite-agent limit, 
which is analogous to a 
\textbf{thermodynamic limit}
(Assumption \ref{a1} and Theorem \ref{T1}). 
The ``thermodynamic" procedure is the core of our framework.  
Then in section \ref{mort}, 
we introduce the ``lifetime" of an insured, modeled by the 
killing time of an associated diffusion process (Assumption \ref{a3}). 
The thermodynamic probability measure 
is then compensated by 
the killing rate (Theorem \ref{T27}). 
The basic framework with the lifetime
is used in section \ref{DTInsurance}
to model the cash-flow of a whole-life insurance with a level premium. 
Its continuous-time version 
is presented in section \ref{cont}, 
where, after taking the ``thermodynamic limit", the premium can be obtained 
by calculating the Laplace transform 
of expectations with respect to
the heterogeneity distribution
(see equation \eqref{pxstar2} in Theorem \ref{Th3.5}). 
The Laplace transform method 
is another core of our framework
and enables computational tractability. 

Then, in section \ref{seckaiyaku}, 
we add to the introduced model the surrender time, 
which is again modelled by a killing time
(Assumption \ref{a4}). 
Under the continuous-time thermodynamic model, 
the pay-off, which is the
revenue minus expenditure (of the 
insurance company), is modelled
by the Laplace transform 
of \eqref{core},
after deriving the thermodynamic limit
under the extended setting 
in Theorem \ref{Th3.2}. 

In sections \ref{BMM} and \ref{Bessel2}, we introduce two specific models where 
we can calculate 
the formula \eqref{core} analytically. 
In the former, we rely on 
the local constancy of the killing rates, which can however approximate fairly general rates. 
If the approximated rate is highly non-linear, then we need to resort to 
the numerical algorithm proposed in 
Theorem \ref{Th4.2}, which is among the 
mathematical 
contributions of the present paper. 
In contrast, 
the latter model gives us completely an analytical expression of \eqref{core}
using the symmetry of the 2-dimensional Bessel process.

\ 

There are some studies analysing the 
surrender risk in the spirit of quantitative finance, such as
\cite{LeCNak} and more recently \cite{Ballotta}, which use continuous time processes similar to ours; however, these are not concerned with heterogeneity. 
They model the lapses by ``jumps", 
that is, exogenous events. 
Such an approach might be called a
``reduced form approach".
From this point of view, our model can be understood as a {\em structural} version
of, for example, the model proposed by O. Le Courtois and H. Nakagawa \cite{LeCNak}
(see Remark \ref{comonCN} in Section \ref{seckaiyaku} below).

\section*{Acknowledgments}
The authors would like to thank Corina Constantinescu (Univ. Liverpool) and Gregory Markowsky (Monash Univ.) for review and for valuable comments and suggestions.

\section{Model without Surrender Risk} \label{secsibou}
\subsection{Basic Model}\label{secsibou1}

Let $X_t$ be a time-homogeneous diffusion process in $ \mathbf{R}^d $, 
\begin{equation}\label{Markov0}
( \{ X_t : t \in \mathbf{R}_+ \}, \{\mathbf{P}^x : x \in \mathbf{R}^n \} ), 
\quad \mathbf{P}^x(X_0=x)=1.
\end{equation}
Here, we consider $X$ to be a
quantified personal profile 
such as health, economic state, or other conditions of a person, which depend on a time parameter $t \in \mathbf{R}_+$. 
The infinitesimal generator of $X$ will be denoted by $ \mathcal{L} $. 
\begin{ass}\label{a0}
We assume that the transition probability of $X$ has a smooth density; that is,
there exists a smooth function $ q $ such that
\begin{equation}
\mathbf{P}^x (X_t \in A) = \int_A q(t,x,y) \,dy, \quad (A \in \mathcal{B}(\mathbf{R}^d)). 
\end{equation}
We further assume that 
$\mathbf{P}^x (X_t \in A)$ is a continuous function 
in $ x $ 
for any fixed $ t > 0 $ and $ A \in \mathcal{B} (\mathbf{R}^d) $. 
\end{ass}
 
Our model for life insurance is as follows. 
Let $N \in \mathbf{N}$ and $\mathcal{I}=\{1,\cdots,N\}$ be the number and the set of the initial insured, respectively. 
The initial personal profile of each insured is $x^i \in \mathbf{R}^d $ for $i \in \mathcal{I}$.
We assume that each of $X^i ~(i=1,\cdots,N)$ 
is distributed as $\mathbb{P}^{x^i}$. 
We also assume that the initial profile of each insured takes values only in
\begin{equation*}
 N^{-1} \mathbf{Z}^d \equiv \{  k/N : k \in \mathbf{Z}^d \}
\end{equation*}

We define a probability measure $\mu_N$ on $ \mathbf{R}^d $ whose support is  $ N^{-1} \mathbf{Z}^d $ by
\begin{equation}\label{init distribution}
\mu_N(A):=\frac{1}{N}\sharp\{i \in \mathcal{I}:x^i \in A \} \ \text{for} \ A \in \mathcal{B}(\mathbf{R}^d).
\end{equation}
We call it the {\em initial distribution measure}. 

\begin{ass}\label{a1}
We assume 
that 
there exists $f \geq 0$ with $\|f\|_1=1$ such that
\begin{equation}\label{measure convergence}
\lim_{N\to \infty}\int_{\mathbf{R}^d} h\left(\dfrac{x}{N}\right) \mu_N\left(dx \right)=\int_{\mathbf{R}^d} h(x)f(x)dx 
\end{equation}
for any bounded continuous function $h$.
\end{ass}

Let a random counting measure $v^N $ be defined by
\begin{equation*}
v^N (t,A) := \frac{1}{N} \sharp \{i \in \mathcal{I} : X^i_t \in A \}, \quad
(A \in \mathcal{B}(\mathbf{R}^d)). 
\end{equation*}
This expresses the proportion of the insured 
whose profile is in $ A $ at time $t$. 

\begin{thm}\label{T1}
Under {\rm Assumptions \ref{a0}} and {\rm \ref{a1}}, we have that 
\begin{equation*}
\lim_{N \to \infty} \mathbb{E} [ v^N (t,A) ]
= \int_A \mathbb{E} [ f(X^*_t) |X^*_0 =x ] \,dx,
\end{equation*}
where $ X^* $ is the adjoint process of $ X $, that is, 
a diffusion process whose 
infinitesimal generator is $ \mathcal{L}^* $,
the adjoint operator of $\mathcal{L}$ with respect to
$ L^2 (\mathbf{R}^d, \mathrm{Leb}) $.
\end{thm}
\begin{proof}
Observe that 
\begin{equation*}
\begin{split}
\mathbb{E}[v^N(t,A)]&=\mathbb{E}\left[\sum_{j \in \mathcal{I}}
\frac{1}{N}1_{\{i \in \mathcal{I} :X_{t}^i \in A\}}(j) \right]\\
&=\frac{1}{N}\sum_{i \in \mathcal{I}}\mathbb{P}(X_{t}^i \in A)=\frac{1}{N}\sum_{i \in \mathcal{I}}\mathbf{P}^{x^i}(X_{t} \in A)\\
&=\frac{1}{N}\sum_{k \in \mathbf{Z}^d }\mathbf{P}^{\frac{k}{N}}(X_{t} \in A)\mu_N(\{k/N\})N\\
&=\sum_{k \in \mathbf{Z}^d }\mathbf{P}^{\frac{k}{N}}(X_{t} \in A)\mu_N(\{k/N\}). 
\end{split}
\end{equation*}
Letting $ N \to \infty $, 
we get 
\begin{equation*}
\begin{split}
\lim_{N \to \infty} \mathbb{E} [ v^N (t,A) ]=\int_{\mathbf{R}^d} \mathbf{P}^x(X_{t} \in A)f(x)dx
\end{split}
\end{equation*}
by using \eqref{measure convergence} in Assumption \ref{a1}.
Moreover, this formula can be rewritten 
(where we can use Fubini's theorem since both $\mathbb{P}^x(X_{t} \in A)$ and $f(x)$ are positive functions) as  
\begin{equation*}
\begin{split}
\lim_{N \to \infty} \mathbb{E} [ v^N (t,A) ] &= \int \mathbf{P}^y (X_t \in A) f(y) dy \\
       &= \int \! \! \left( \int  q(t,y,x) 1_A (x) dx \right) f(y) dy \\
       &= \int 1_A(x) dx \int q (t,y,x) f(y) dy.\\
\end{split}
\end{equation*}
Therefore, $ A \mapsto u(t,A) $ is an absolutely continuous measure with respect to Lebesgue measure, and its density is given by
\begin{equation*}
u(t,x) := \int q (t,y,x) f(y) dy .
\end{equation*}
Using the adjoint  process (with respect to Lebesgue measure) 
$ X^* $, the function $ u $ can be expressed as
\begin{equation}  \label{ninzuu}
u(t,x) = \mathbb{E} [f( X^*_t) |X^*_0 = x ]
=: \mathbf{E}^x [f(X^*_t)],
\end{equation}
where we abuse the notation 
$ \mathbf{E}^x $ a bit. 
\end{proof}

\subsection{Mortality Model}\label{mort}
We introduce a mortality model, where the force of mortality is dependent only on the current personal profile.
In the model, 
the insured do not surrender the policy. 

Let $\hat{X}$ be a killed process 
obtained from the Markov process \eqref{Markov0} defined 
with a random time $ \zeta $ as follows,
\begin{equation}\label{Markov1}
\hat{X}_t  = X_t 1_{ \{ \zeta > t \} } + \infty 1_{ \{ \zeta \leq t\} }.
\end{equation}
\begin{ass}\label{a3}
We assume that 
\begin{equation}\label{a1f}
\mathbb{P} ( \zeta > t \,| \,\sigma (X_s : s \leq t)  ) = e^{- \int_0^t V(X_s)\,ds }
\end{equation}
and
\begin{equation}
\lim_{t \to \infty}\mathbb{P} ( \zeta > t \,| \,\sigma (X_s : s \leq t)  )=0 \label{sibou}
\end{equation}
with a positive function $ V $.
\end{ass}

\begin{rem}
The mortality in our model is consistent with
the ones \cite{VMS}, and \cite{WM}, \cite{YMV} in demography.  
\end{rem}
Then, 
we know that (see e.g. \cite[chapter III, Theorem 18.6]{RW}) the 
transition semigroup $ (\hat{T}_t)_{t \geq 0} $ of $ \hat{X} $ is given by
\begin{equation*}
\begin{split}
\hat{T}_t f (x) &:= 
\mathbb{E} [ f (\hat{X}_s) | X_0 = x ] \\
&= \mathbb{E} [ f (X_s) e^{-\int_0^t V(X_s) \,ds} | X_0 = x ],
~~( f \in C_0 (\mathbf{R}^d)).
\end{split}
\end{equation*}
By the Feynman-Kac formula (see e.g. \cite[chapter III 19.]{RW}), 
the infinitesimal generator $ \hat{\mathcal{L}} $ of the semigroup $\hat{T}$ is given, using the infinitesimal generator $\mathcal{L}$ of $T$, as follows
\begin{equation*}
\hat{\mathcal{L}} f(x)  = \mathcal{L} f(x) - V(x) f(x), 
\quad ( f \in \mathscr{D} (\mathcal{L}) ), 
\end{equation*}
where 
the domain of $ \mathcal{L} $
is, as usual, 
the space of such $ f $ that 
\begin{equation*}
    \lim_{t\to 0} \frac{\hat{T}_t f -f}{t}
\end{equation*}
exists in $ C_0 (\mathbf{R}^d)$. 
We will sometimes use the notation 
\begin{equation*}
    \hat{T}_t f (x)
    = \mathbf{E}^x[f(\hat{X})]= \mathbf{E}^x[f({X})
    e^{-\int_0^t V(X_s) \,ds}]
\end{equation*}
for the purpose of clarifying 
both the starting point and the sample path which we are looking at, even though this may again be a bit of an abuse of notation. 
\begin{ass}\label{that}
We assume that 
$\hat{T}_t $ has a density, that is,
there exists a smooth function $ q_V $ such that
\begin{equation*}
\hat{T}_t f (x) = \int_{\mathbf{R}^d} q_V (t,x,y) \,f(y) \,dy.
\end{equation*}
\end{ass}

Let the model size be $ N $ as in 
the preceding section, 
and the initial 
distribution measure $ \mu^N $ 
be as \eqref{init distribution}. 
Let a random counting measure $v^N$ be redefined by
\begin{equation*}
v^N (t, A) = \frac{1}{N}\sharp \{ i \in \mathcal{I} : \hat{X}^i_t \in A \}, \quad
(A \in \mathcal{B}(\mathbf{R}^d)). 
\end{equation*}

\begin{thm}\label{T27}
Let $ A \in \mathcal{B}(\mathbf{R}^d) $. 
Under Assumptions  \ref{that} and {\rm \ref{a1}}, we have that 
\begin{equation}\label{VXstar}
\lim_{N \to \infty} \mathbb{E} [ v^N (t,A) ]
= \int_A \mathbb{E} [ f(X^*_t)e^{-\int_0^tV(X^*_s)ds} |X^*_0 =x ] \,dx.
\end{equation}
\end{thm}

\begin{proof}
Let us calculate $\mathbb{E} [ v^N (t,A) ]$ as
\begin{equation*}
\begin{split}
\mathbb{E}[v^N(t,A)]&=\mathbb{E}\left[\sum_{j \in \mathcal{I}}
\frac{1}{N}1_{\{i \in \mathcal{I}:\hat{X}_{t}^i \in A\}}(j) \right]\\
&=\frac{1}{N}\sum_{i \in \mathcal{I}}\mathbb{P}(\hat{X}_{t}^i \in A)=\frac{1}{N}\sum_{i \in \mathcal{I}}\mathbf{P}^{x^i}(\hat{X}_{t} \in A)\\
&=\sum_{k \in \mathbf{Z}^d }\mathbf{P}^{k/N}(\hat{X}_{t} \in A)\mu_N(\{k/N\}).\\
\end{split}
\end{equation*}
Letting $N \to \infty$, we get the formula
\begin{equation*}
\begin{split}
\lim_{N \to \infty} \mathbb{E} [ v^N (t,A) ]=\int_{\mathbf{R}^d} \mathbf{P}^{y}(\hat{X}_{t} \in A)(y) f(y) dy
\end{split}
\end{equation*}
by using \eqref{measure convergence}
in Assumption \ref{a1}.
Then, by the same 
procedure as we used in the proof of Theorem \ref{T1}, we get \eqref{VXstar}.
\if1
Moreover, this formula can be rewritten, by using Fubini's theorem since both $\hat{T}1_A(k/N)$ and $f(x)$ are positive functions,
as
\begin{equation*}
\begin{split}
\lim_{N \to \infty} \mathbb{E} [ v^N (t,A) ] &= \int \hat{T}_t (1_A) (y) f(y) dy \\
&= \int \! \! \left( \int  q_V (t,y,x) 1_A (x) dx \right) f(y) dy \\
&= \int 1_A(x) dx \int q_V (t,y,x) f(y) dy .
\end{split}
\end{equation*}
Therefore, $ A \mapsto u(t,A) $ is a measure that is absolutely continuous
with respect to Lebesgue measure, and its density is given by
\begin{equation*}
u(t,x) := \int q_V (t,y,x) f(y) dy .
\end{equation*}
Using the adjoint Markov process (for  Lebesgue measure) 
$ X^* $, the function $ u $ can be expressed as 
\begin{equation*} 
u(t,x) =  \mathbb{E}^x [f( X^*_t) e^{- \int_0^t V ( X^*_t) \,dt } ].
\end{equation*}
\fi
\end{proof}


\subsection{Discrete-Time Level-Premium Insurance Model}\label{DTInsurance}
We consider 
an insurance 
where both the premium and the insurance proceeds are paid
discretely 
at time $ t = 0,1,\cdots $. 
Let $p$ and $A$ be the (level) premium and the sum insured, respectively. 
Then, the revenue of the insurance company $ v_t(p) $ at each time $t$ is given by
\begin{equation*}
\begin{split}
v_t (p) &= p \cdot \sharp \{ i \in \mathcal{I} : \hat{X}^i_t \ne \infty \}
\quad (t=0,1, 2 ,...).
\end{split}
\end{equation*}

In our model, the expenditure of the insurance company $c_t$ at each time t is given by
\begin{equation*}
\begin{split}
&c_0 = 0, \\
& c_t = A \cdot \sharp \{ i \in \mathcal{I} :  \hat{X}^i_t = \infty ,\hat{X}^i_{t-1} \ne \infty\} 
\quad (t= 1, 2,...).
\end{split}
\end{equation*}
Let the expected return 
$ R_d $
of the discrete-time insurance model at time $0$ be defined by
\begin{equation}\label{RNnp}
R_d (N,p):=\sum_{t=0}^\infty e^{-rt}\mathbb{E}[v_t(p)-c_t],
\end{equation}
and $p_d (N)$ be the 
solution of $ R_d (N,p;\mathbf{N}) = 0 $.
Note that the solution is unique 
and 
strictly positive since 
\eqref{RNnp} is linear in $ p $ 
and we have clearly
\begin{equation*}
\sum_{t=0}^\infty e^{-rt}
\mathbb{P} ( i \in \mathcal{I} : \hat{X}^i_t \ne \infty )
>0 
\end{equation*}
and
\begin{equation*}
A \sum_{t=0}^\infty e^{-rt}
\mathbb{E}[ \sharp \{ i \in \mathcal{I} :  \hat{X}^i_t = \infty ,\hat{X}^i_{t-1} \ne \infty\} ]
>0
\end{equation*}
by \eqref{sibou}
in Assumption \ref{a3}. 
\begin{thm} \label{premium1}
We have that
\begin{equation*}
\begin{split}
p_d(N)= 
\frac{\displaystyle A\sum_{t=1}^\infty e^{-rt}\sum_{k \in \mathbf{Z}^d }\mathbf{E}^{\frac{k}{N}}[e^{-\int_0^{t-1} V(X_s)ds}-e^{-\int_0^t V(X_s)ds}]\mu_N(\{k/N\})}
{\displaystyle \sum_{t=0}^\infty e^{-rt} \sum_{k \in \mathbf{Z}^d }\mathbf{E}^{\frac{k}{N}}[e^{-\int_0^t V(X_s)ds}]\mu_N(\{k/N\})}.
\end{split}
\end{equation*}
Moreover, $p_d (\infty)$ defined by
\begin{equation*}
p_d (\infty):=\lim_{N \to \infty}p_d (N)
\end{equation*}
is given by
\begin{equation*}
\begin{split}
&p_d(\infty)= 
\frac{\displaystyle A\sum_{t=1}^\infty e^{-rt}\int_{\mathbf{R}^d} \mathbf{E}^{y}[e^{-\int_0^{t-1} \tilde{V}(X^*_s)ds}f(X^*_{t-1})-e^{-\int_0^t \tilde{V}(X^*_s)ds}f(X^*_t)]dy}
{\displaystyle \sum_{t=0}^\infty e^{-rt} \int_{\mathbf{R}^d} \mathbf{E}^{y}[e^{-\int_0^t \tilde{V}(X^*_s)ds}f(X^*_t)]dy}, 
\end{split}
\end{equation*}
using the adjoint process. 
\end{thm}

\begin{proof}
We first obtain the expected return $ R_d (N,p) $ at time $0$.
The expected revenue at time $0$ is 
\begin{align*}
\sum_{t=0}^\infty e^{-rt}\mathbb{E}[v_t(p)] 
=&\sum_{t=0}^\infty e^{-rt}\mathbb{E} 
\left[ p \cdot 1_{\{i \in \mathcal{I}: \hat{X}_t^i \ne \infty \}}\right] \\
=&p\sum_{t=0}^\infty e^{-rt} \sum_{i \in \mathcal{I}}\mathbb{P}(\hat{X}_{t}^i \ne \infty)\\
=&p\sum_{t=0}^\infty e^{-rt} \sum_{k \in \mathbf{Z}^d }\mathbf{P}^{\frac{k}{N}}(\hat{X}_{t} \ne \infty)\mu_N(\{k/N\})N.
\end{align*}
Since $ \{ \hat{X}_t \ne \infty \}=
\{ \zeta > t \} $, the expected revenue is now given by
\begin{align}
&=p\sum_{t=0}^\infty e^{-rt} \sum_{k \in \mathbf{Z}^d }\mathbf{P}^{\frac{k}{N}}(\zeta>t)\mu_N(\{k/N\})N \nonumber\\
&=p\sum_{t=0}^\infty e^{-rt} \sum_{k \in \mathbf{Z}^d }
\mathbf{E}^{\frac{k}{N}}[\mathbf{E}^{\frac{k}{N}}[1_{\{\zeta>t\}}|\sigma(X_s:s\leq t)]]\mu_N(\{k/N\})N\nonumber\\
&=p\sum_{t=0}^\infty e^{-rt} \sum_{k \in \mathbf{Z}^d }\mathbf{E}^{\frac{k}{N}}[e^{-\int_0^t V(X_s)ds}]\mu_N(\{k/N\})N \label{revenue1}.
\end{align}
Next, we calculate the expected expenditure at time $0$ as
\begin{align*}
\mathbb{E}[c_t] 
=&
A\mathbb{E}\left[ 1_{\{i \in \mathcal{I}: \hat{X}_t^i=\infty,\hat{X}_{t-1}^i\ne\infty \}}\right] \\
=&A
\sum_{i \in \mathcal{I}}\mathbb{P}(\hat{X}_{t}^i =\infty , \hat{X}_{t-1}^i \ne \infty) \\
=&A
\sum_{k \in \mathbf{Z}^d }\mathbf{P}^{\frac{k}{N}}(\hat{X}_{t} =\infty , \hat{X}_{t-1} \ne \infty)\mu_N(\{k/N\})N.
\end{align*}
As above we can use
expressions with
$ \zeta $ instead, and 
the expected expenditure is now 
\begin{align}
&=A
\sum_{k \in \mathbf{Z}^d }\mathbf{P}^{\frac{k}{N}}(t-1<\zeta \leq t)\mu_N (\{k/N\})N\nonumber \\
&=A
\sum_{k \in \mathbf{Z}^d }
\left(\mathbf{E}^{\frac{k}{N}}[
\mathbf{E}^{\frac{k}{N}}[1_{\{\zeta>t-1\}}
-
1_{\{\zeta>t\}}\mid \sigma(X_s:s\leq t)]]\right)
\mu_N(\{k/N\}) N\nonumber\\
&=A
\sum_{k \in \mathbf{Z}^d }\mathbf{E}^{\frac{k}{N}}[e^{-\int_0^{t-1} V(X_s)ds}-e^{-\int_0^t V(X_s)ds}]\mu_N(\{k/N\})N\label{expenditure1}.
\end{align}
Using the formulas \eqref{revenue1} and \eqref{expenditure1}, the expected return $ R_d (N,p) $ is calculated as follows:
\begin{equation}
\begin{split}
&R_d(N,p)\\
=&\sum_{t=0}^\infty e^{-rt}\mathbb{E}[v_t(p)-c_t]\\
=&p\sum_{t=0}^\infty e^{-rt} \sum_{k \in \mathbf{Z}^d }\mathbf{E}^{\frac{k}{N}}[e^{-\int_0^t V(X_s)ds}]\mu_N(\{k/N\})N \\
&\quad -A\sum_{t=1}^\infty e^{-rt}\sum_{k \in \mathbf{Z}^d }\mathbf{E}^{\frac{k}{N}}[e^{-\int_0^{t-1} V(X_s)ds}-e^{-\int_0^t V(X_s)ds}]\mu_N(\{k/N\})N. \label{expect return1}
\end{split}
\end{equation}
Thus, 
the premium $ p_d (N) $ 
is given by
\begin{equation}
p_d (N)=\frac{\displaystyle A\sum_{t=1}^\infty e^{-rt}\sum_{k \in \mathbf{Z}^d }\mathbf{E}^{\frac{k}{N}}[e^{-\int_0^{t-1} V(X_s)ds}-e^{-\int_0^t V(X_s)ds}]\mu_N(\{k/N\})}
{\displaystyle \sum_{t=0}^\infty e^{-rt} \sum_{k \in \mathbf{Z}^d }\mathbf{E}^{\frac{k}{N}}[e^{-\int_0^t V(X_s)ds}]\mu_N(\{k/N\})}. \label{p1}
\end{equation}
Letting $ N \to \infty $, 
which is possible since each term 
in \eqref{p1} converges by Assumption \ref{a1},
we get the following formula: 
\begin{equation*}
p_d (\infty)=
\frac{\displaystyle A\sum_{t=1}^\infty e^{-rt}\int_{\mathbf{R}^d} \mathbb{E}^{x}[e^{-\int_0^{t-1} V(X_s)ds}-e^{-\int_0^t V(X_s)ds}]f(x)dx}
{\displaystyle \sum_{t=0}^\infty e^{-rt} \int_{\mathbf{R}^d} \mathbb{E}^{x}[e^{-\int_0^t V(X_s)ds}]f(x)dx}.
\end{equation*}
Moreover, this formula can be rewritten by using the adjoint process and the potential $ \tilde{V} $ of $\mathcal{L}^*$ as follows:
\begin{equation*}
\begin{split}
& p_d (\infty)= \\
&\frac{\displaystyle A\sum_{t=1}^\infty e^{-rt}\int_{\mathbf{R}^d}\mathbb{E}^{y}[e^{-\int_0^{t-1} \tilde{V}(X^*_s)ds}f(X^*_{t-1})-e^{-\int_0^t \tilde{V}(X^*_s)ds}f(X^*_t)]dy}
{\displaystyle \sum_{t=0}^\infty e^{-rt} \int_{\mathbf{R}^d}\mathbb{E}^{y}[e^{-\int_0^t \tilde{V}(X^*_s)ds}f(X^*_t)]dy}.
\end{split}
\end{equation*}
\end{proof}

\subsection{Continuous-Time Level-Premium Insurance Model}\label{cont}
We will consider a continuous-time payment model
in this section.
The revenue of the insurance company during $ [0,t] $
is given by
\begin{equation*}
    \int_0^t e^{-r s}
    v_s (p) \,ds,
\end{equation*}
where 
\begin{equation*}
\begin{split}
v_t (p) &= p \cdot \sharp \{ i \in \mathcal{I} : \hat{X}^i_t \ne \infty \},
\end{split}
\end{equation*}
while 
the expenditure of the insurance company during $[0,t]$ is given by 
\begin{equation*}
\begin{split}
& C_t = \sum_{j \in \mathcal{I}} A e^{-r \zeta^j} 1_{\{i \in \mathcal{I}: \zeta^i<t \}}(j).
\end{split}
\end{equation*}
Note that $ C_t $ is increasing and bounded so that $ c_\infty $ exists is and is finite almost surely. In fact, it is bounded by $ A N $.

Let the expected return 
of the continuous time model at time $0$ 
be defined by
\begin{equation*}
R_c (N,p):=\int_{0}^\infty e^{-rt}\mathbb{E}[v_t(p)]dt
- \mathbb{E}[C_\infty],
\end{equation*}
and $p_c (N)$ be the 
solution of $R_c (N, p ) = 0 $. 
Note that the solution is unique 
and 
strictly positive
by the same reasoning 
as in Section \ref{DTInsurance}.
\begin{thm}\label{Th3.5}
We have that 
\begin{align}
p_c (N)=\frac{\displaystyle A\int_{0}^\infty e^{-rt}\sum_{k \in \mathbf{Z}^d }\mathbf{E}^{\frac{k}{N}}[V(X_s)e^{-\int_0^{t} V(X_s)ds}]\mu_N(\{k/N\})dt}
{\displaystyle \int_0^\infty e^{-rt} \sum_{k \in \mathbf{Z}^d }\mathbf{E}^{\frac{k}{N}}[e^{-\int_0^{t} V(X_s)ds}]\mu_N(\{k/N\})dt}.
\end{align}
Moreover, $p_c(\infty)$ defined by
\begin{equation*}
p_c (\infty):=\lim_{N \to \infty}p_c(N)
\end{equation*}
is given by
\begin{align}\label{pxstar2}
p_c(\infty)=\frac{\displaystyle  A\int_0^\infty e^{-rt}\int_{\mathbf{R}^d}
\mathbf{E}^y
[f(X^*_t)\tilde{V}(X^*_t)e^{-\int_0^{t} \tilde{V}(X^*_s)ds}]dydt}
{\displaystyle \int_0^\infty e^{-rt} \int_{\mathbf{R}^d}\mathbf{E}^{y}[f(X^*_t)e^{-\int_0^{t} \tilde{V}(X^*_s)ds}]dydt},
\end{align}
where $\tilde{V}$ is the potential of $\mathcal{L}^*$.
\end{thm}

\begin{proof}
We first obtain the expected return $ R_c (N,p)$.
Calculating $v_t(p)$ like we did in \eqref{revenue1}, we get 
\begin{align}\label{revenue2}
\int_{0}^\infty e^{-rt}\mathbb{E}[v_t(p)]dt=p\int_0^\infty e^{-rt} \sum_{k \in \mathbf{Z}^d }\mathbf{E}^{\frac{k}{N}}[e^{-\int_0^{t} V(X_s)ds}]\mu_N(\{k/N\})N dt,
\end{align}
while the expected expenditure at time 0:,
\begin{align*}
\mathbb{E}[C_t]&= \mathbb{E}\left[\sum_{j \in \mathcal{I}} A e^{-r \zeta^j} 
1_{\{i\in \mathcal{I}: \zeta^i<t \}}(j) \right],
\end{align*}
can be, by Lemma \ref{appendlem} in Appendix \ref{continuous death}, rewritten as
\begin{align*}
A \int_0^t e^{-rs}\sum_{i \in \mathcal{I}} \mathbb{E}[V(X^i_s)e^{-\int^s_0V(X^i_u)du}]ds.
\end{align*}
Then, as we did in \eqref{expenditure1}, we get the following formula:
\begin{align} \label{expenditure2}
\mathbb{E}[C_t]=A \int_0^t e^{-rs} \sum_{k \in \mathbf{Z}^d }  \mathbf{E}^{\frac{k}{N}}[V(X_s)e^{-\int^s_0V(X_u)du}]\mu_N(\{k/N\})N ds.
\end{align}
Using the formulas \eqref{revenue2} and \eqref{expenditure2}, the expected return $ R_c (N,p)$ is obtained as 
\begin{align}
&R_c (N,p) \nonumber \\
=&\int_{0}^\infty e^{-rt}\mathbb{E}[v_t(p)]dt-
\mathbb{E}[C_\infty] \nonumber\\
=&p\int_0^\infty e^{-rt} \sum_{k \in \mathbf{Z}^d }\mathbf{E}^{\frac{k}{N}}[e^{-\int_0^{t} V(X_s)ds}]\mu_N(\{k/N\})N dt\nonumber \\
&\quad -A\int_0^\infty e^{-rt}\sum_{k \in \mathbf{Z}^d } \mathbf{E}^{\frac{k}{N}}[V(X_t)e^{-\int^t_0V(X_u)du}]\mu_N(\{k/N\})N dt. \label{expect return2}
\end{align}
Thus, 
the premium $ p(N) $ is given by
\begin{align}
p_c (N)=\frac{\displaystyle A \int_0^\infty e^{-rt}\sum_{k \in \mathbf{Z}^d }  \mathbf{E}^{\frac{k}{N}}[V(X_t)e^{-\int^t_0V(X_u)du}]\mu_N(\{k/N\})dt}
{\displaystyle \int_0^\infty e^{-rt} \sum_{k \in \mathbf{Z}^d }\mathbf{E}^{\frac{k}{N}}[e^{-\int_0^{t} V(X_s)ds}]\mu_N(\{k/N\})dt}. \label{p2}
\end{align}
Then, by the same 
procedure as we used in the proof of Theorem \ref{premium1}, we get \eqref{pxstar2}.
\if2
Taking the limit as $N$ tends to $\infty$ with respect to \eqref{p2}, we get the following formula since each term converges by Assumption \ref{a1}
\begin{align*}
p(\infty)=\frac{\displaystyle A\int_{0}^\infty e^{-rt}\int_{-\infty}^{\infty}\mathbb{E}^{x}[V(X_t)e^{-\int_0^{t} V(X_s)ds}]f(x)dxdt}
{\displaystyle \int_{0}^\infty e^{-rt} \int_{-\infty}^{\infty}\mathbb{E}^{x}[e^{-\int_0^{t} V(X_s)ds}]f(x)dxdt}.
\end{align*}
Moreover, this formula can be rewritten by using the adjoint process as follows,
\begin{align*}
p(\infty)=\frac{\displaystyle A\int_{0}^\infty e^{-rt}\int_{-\infty}^{\infty}\mathbb{E}^{y}[f(X^*_t)V(X^*_t)e^{-\int_0^{t} V(X^*_s)ds}]dydt}
{\displaystyle \int_{0}^\infty e^{-rt} \int_{-\infty}^{\infty}\mathbb{E}^{y}[f(X^*_t)e^{-\int_0^{t} V(X^*_s)ds}]dydt}.
\end{align*}
\fi
\end{proof}


\section{Model with Surrender Risk}\label{seckaiyaku}
\subsection{Mortality Model with Surrender Risk}
In this section 
we consider a model 
where insurers can surrender their policy 
by 
balancing their personal conditions with
the premium. 

\if2
Let $Z(p)$ be a killed process 
obtained from the diffusion process 
$ X $ of \eqref{Markov0}
with Assumption \ref{a0}
and the random time $ \zeta $ with Assumption \eqref{a3}, 
defined 
with a cancellation time $ \xi (p) $, as follows,
\begin{equation}
Z_t (p) = X_t 1_{ \{ \xi(p) > t \} } 
1_{\{ \zeta > t \} } + \infty 1_{ \{ \zeta \leq t \} } 
- \infty 1_{ \{ \xi(p) \leq t \} } 1_{ \{ \zeta > t \}}.
\end{equation}
\fi
In addition to $ \zeta $, we introduce a new random time $ \xi (p) $, which is dependent on a parameter $ p $, 
satisfying the following assumption.
\begin{ass}\label{a4}
We assume that 
\begin{enumerate}
\item \begin{equation*}
\mathbb{P} ( \xi(p) > t \,| \,\sigma (X_s : s \leq t) ) = 
e^{- \int_0^t D(p,X_s)\,ds }
\end{equation*}
with a positive measurable function $ D 
:[0,\infty) \times \mathbf{R}^d \ni (p,x)  \mapsto D(p,x) \in \mathbf{R}_{>0} $,
which is increasing in $ p $
and decreasing in $ x $. 
\item The random times $ \zeta $ and $ \xi(p) $ are conditionally independent in the following sense:
\begin{equation}\label{conditional expectation}
\begin{split}
&\mathbb{P} ( \xi(p) > t , \zeta > t \,| \,\sigma (X_s : s \leq t) ) \\
=& 
\mathbb{P} ( \xi(p) > t \,| \,\sigma (X_s : s \leq t) )
\mathbb{P} ( \zeta > t \,| \,\sigma (X_s : s \leq t) ) \\
=& e^{-\int_0^t V(X_s)\,ds - \int_0^t D(X_s,p)\,ds }.
\end{split}
\end{equation}
\end{enumerate}
\end{ass}

\subsection{Continuous-time Level-Premium Insurance Model with Surrender Risks}
The revenue of the insurance company during $ [0,t] $
is given by
\begin{equation*}
    \int_0^t e^{-rs}  v_s (p)\,ds,
\end{equation*}
where 
\begin{equation}
v_t(p) = p \cdot \sharp \{ i \in \mathcal{I} : \zeta^i> t,\xi^i(p)> t \}, \label{income}
\end{equation}
while the expenditure of the insurance company during $ [0,t] $
is given by
\begin{equation}
\begin{split}
& C_t = A \sum_{j \in \mathcal{I}} e^{-r \zeta^j} 1_{\{ i \in \mathcal{I} : ~ \zeta^i \leq t, ~\xi^i(p)> t\}}(j).\label{spending}
\end{split}
\end{equation}
The expected return at time $0$
is the same as the one in the previous section, namely, 
\begin{equation*}
R_s (N,p) =\int_0^\infty e^{-rt}\mathbb{E}[v_t(p)]dt
- \mathbb{E}[C_\infty],
\end{equation*}
where the subscript $ s $ 
is put to indicate that it is the one with surrender risk. 
It should be noted that the expected return is no longer a linear function in 
$ p $, and thus 
we may not have uniqueness of the solution 
$ p $ for $ R_c (N,p) = 0 $.

\begin{thm}\label{Th3.2}
We have that 
\begin{equation}\label{expect return3}
\begin{split}
&R_s (N,p)\\
=&\int_0^\infty e^{-rt}\mathbb{E}[v_t(p)]dt
- \mathbb{E}[C_\infty ]\\
=&\int_0^\infty e^{-rt} \sum_{k \in \mathbf{Z}^d }
\mathbf{E}^{\frac{k}{N}}[
(p-A V(X_t))e^{-\int_0^t V(X_s)+D(X_s,p)ds}]\mu_N(\{k/N\})N dt. \\
\end{split}
\end{equation}
\end{thm}

\begin{proof}
The proof is almost the same
as for Theorem \ref{Th3.5}. First, 
\begin{align*}
&\int_0^\infty e^{-rt}\mathbb{E}[v_t(p)]dt \\
=&\int_0^\infty e^{-rt}
\mathbb{E} 
\left[
p \cdot \sharp \{ i \in \mathcal{I} : \zeta^i> t,\xi^i(p)> t \} \right]dt \\
=&p\int_0^\infty e^{-rt} 
\sum_{i \in \mathcal{I}}\mathbb{P}(\zeta^i>t,\xi^i(p)>t)dt \\
=&p\int_0^\infty e^{-rt} \sum_{k \in \mathbf{Z}^d }\mathbf{P}^{\frac{k}{N}}(\zeta>t,\xi(p)>t)\mu_N(\{k/N\})Ndt \\
=&p\int_0^\infty e^{-rt} 
\sum_{k \in \mathbf{Z}^d }
\mathbf{E}^{\frac{k}{N}}[\mathbb{E}[1_{\{\zeta>t,\xi(p)>t\}}|\sigma(X_s:s\leq t)]]\mu_N(\{k/N\})Ndt 
\end{align*}
This formula can be rewritten, by \eqref{conditional expectation}, as 
\begin{align}
=&p\int_0^\infty e^{-rt} \sum_{k \in \mathbf{Z}^d }\mathbb{E}^{\frac{k}{N}}[e^{-\int_0^t (V(X_s)+D(X_s,p))\,ds}]\mu_N(\{k/N\})Ndt \label{revenue3}.
\end{align}
Next, the expected expenditure at time 0, 
\begin{align*}
\mathbb{E} [C_\infty] 
&=\lim_{t \to \infty}\mathbb{E}\left[A 
\sum_{j \in \mathcal{I}} e^{-r \zeta^j} 
1_{\{ i \in \mathcal{I} : ~ \zeta^i<t, ~\xi^i(p)> t\}}(j)\right],
\end{align*}
is
obtained as 
\begin{align*}
&A\int_0^t e^{-rs} 
\sum_{i \in \mathcal{I}}\mathbb{P}(\zeta^i \in ds, 
\xi^i(p)>t )  \\
& =A\int_0^t
e^{-rs} \sum_{k \in \mathbf{Z}^d }
\mathbf{P}^{\frac{k}{N}}
(\zeta \in ds
,\xi(p)>t)
\mu_N(\{k/N\})N 
\\
&=A\int_0^t e^{-rs}
\sum_{k \in \mathbf{Z}^d }
\mathbf{E}^{\frac{k}{N}}
[\mathbf{E}^{\frac{k}{N}}[1_{\{\zeta \in ds,\ 
\xi(p)>t \}}|\sigma(X_u:u\leq t)]]\mu_N(\{k/N\})N 
\nonumber
\end{align*}
This expression can be shown, by \eqref{conditional expectation}
and Lemma \ref{appendlem} in Appendix \ref{continuous death}, to be equal to 
\begin{align}
&A\int_0^\infty e^{-rt} \sum_{k \in \mathbf{Z}^d }\mathbf{E}^{\frac{k}{N}}[V(X_t)e^{-\int_0^t ( V(X_s)+D(X_s,p)) \,ds}]\mu_N(\{k/N\})N dt.
\label{expenditure3}
\end{align}
Taking the limit of \eqref{expenditure3} as $t\to \infty$, and 
using the formula \eqref{revenue3}, 
we obtain \eqref{expect return3}. 
\end{proof}

We can 
define the virtual average expected return 
as 
\begin{align}\label{core}
\begin{aligned}
&\mathrm{VAR}_s (p) 
:= \lim_{N \to \infty} \frac{1}{N} R_s (N, p) \\
&=\int_0^\infty e^{-rt} \int_{\mathbf{R}^d}\mathbf{E}^{x}[
(p - AV(X_t))e^{-\int_0^t V(X_s)+D(X_s,p)ds}]f(x)dxdt
\\
&= \int_{0}^\infty e^{-rt} \int_{\mathbf{R}^d} \mathbf{E}^{y}[(p- AV(X^*_t)) f(X^*_t)e^{-\int_0^{t} V(X^*_s)+D(X^*_s,p)ds}]dydt.
\end{aligned}
\end{align}
The zero(s) of $ \mathrm{VAR} (p) $ can be a good approximation 
of the zero(s) of 
$ R_s (N,p) $.

\begin{rem}
Since 
$ R_s (N, 0) < 0 $ and $ R_s $ is continuous in $ p $, 
the solution $ p^* $ 
of $ R_S (N, p^*) = 0 $
exists if and only if 
$ R_s (N, p) >0 $ for some $ p $. 
An evident sufficient condition that ensures 
the existence of the solution 
is that 
$ V $ and $ D $ are bounded,
since this implies 
$ \lim_{p \to \infty } R_s (N, p)= +\infty$. 
\end{rem}

\begin{rem}\label{comonCN}
In the paper by O. Le Courtois and H. Nakagawa \cite{LeCNak}, 
the number of surrenders is modeled by 
a counting process, 
and the expected number of the remaining participants at $ t $ is expressed by 
its intensity process, 
while in our framework 
it is given by 
$ \int f(x) \mathbf{E} [ e^{-\int_0^t D(X_s,p)ds} ] dx $ using the ``state process" $ X $. 
\end{rem}

\section{A Model by 
Brownian Motion with Constant Drift}
\label{BMM}

\subsection{Description of the model and an 
expression of the expected return}
In this section, 
we specifically assume that the personal condition process $ X_t $ is one dimensional Brownian motion with drift; that is,
\begin{equation*}
 X_t = a W_t + bt,
\end{equation*}
where $W_t$ is standard Brownian motion, $a > 0,\ b \in \mathbf{R}$,
so that 
\begin{equation*}
P^x (X_t \in A) = \int_A \frac{1}{\sqrt{2 \pi a^2 t}} 
e^{-\frac{(y-x-bt)^2}{2a^2 t}} \,dy ,
\quad A \in
\mathcal{B}(\mathbf{R}).
\end{equation*}
Clearly, 
Assumption \ref{a0} is satisfied. 

Moreover, we assume that 
the killing rate functions 
$ V(y) $ and $D(y,p)$ 
are step functions; that is, 
\begin{equation*}
\begin{split}
 & V(y) = \sum_{i=1}^{M} \lambda_i 1_{ (y_{i-1},  y_i]  }(y) , \\
 & D(y,p) = \sum_{i=1}^{M} \mu_i(p) 1_{ ( y_{i-1},   y_i ]}(y) , \\
 & ( M \in \mathbf{N} , \lambda_i, \mu_i(p) \in \mathbf{R},
              -\infty = y_0 < y_1 < ... < y_M = \infty ).
\end{split}
\end{equation*}
The model is so designed 
that we can fit/calibrate any data
to a certain extent 
provided that the personal condition is expressed by a real number.

In this case, 
the expected return 
$ R_s (N, p) $ can be 
calculated,
but only if we 
consider the inversion of a large scale 
matrix to be tractable.
Below we first show 
how the expected return is calculated out of an inversion 
of a large matrix. 
Then, we propose a numerical scheme 
to reduce the burden. 

We define $ z_V $ and $z_1$ as follows: for $ y \in \mathbf{R}$, 
\begin{equation*}
\begin{split}
&z_V(y) := \int_0^\infty e^{-r t} E^{y}[V(X_t)e^{- \int_0^{t} (V(X_s) + D(X_s,p))\,ds } ] \,dt, \\
\end{split}
\end{equation*}
and 
\begin{equation*}
\begin{split}
&z_1(y) := \int_0^\infty e^{-r t} E^{y}[e^{- \int_0^{t}(V(X_s) + D(X_s,p))\,ds }] \,dt.
\end{split}
\end{equation*}
Then we can express the expected return as 
\begin{equation*}
    \begin{split}
        &R_s (N,p)
        = \sum_{k=-\infty}^\infty 
        N \mu_N \left(\frac{k}{N} \right) 
        \left(p u_1 \left(\frac{k}{N} \right)
        -A u_V \left(\frac{k}{N} \right) \right).
    \end{split}
\end{equation*}

\begin{lem}
We have that 
\begin{equation*}
\begin{split}
z_V(y) = \sum_{i=1}^{M}
        &  \Bigl\{ C_{i,+}e^{\alpha_{i,+} y}
         + C_{i,-}e^{\alpha_{i,-}y} 
      + \gamma_i  
      \Bigr\} \ 1_ { \{ y_{i-1} < y <  y_i \} },
      \end{split}
\end{equation*}
and
\begin{equation*}
\begin{split}
z_1(y) = \sum_{i=1}^{M}
        &  \Bigl\{ \widetilde{C}_{i,+}e^{\alpha_{i,+} y}
         + \widetilde{C}_{i,-}e^{\alpha_{i,-}y} +
         \widetilde{\gamma_i} 
         \Bigr\}
          \ 1_{ \{ y_{i-1} < y < y_i \} },
\end{split}
\end{equation*}
where, for $i=1, \cdots, M $, 
\begin{equation*}
    \alpha_{i, \pm } = \frac{-b \pm \sqrt{b^2+2a^2(\lambda_i+\mu_i(p)+r)}}{a^2}, 
\end{equation*}
\begin{equation*}
    \gamma_i = \frac{\lambda_i}{2(\lambda_i+\mu_i(p)+r)} = \lambda_i \tilde{\gamma_i}.
\end{equation*}
The constants $ C_{i,\pm},\widetilde{C}_{i,\pm} \in \mathbf{R} $ are given in the following way: 
$ C_{1,-} = C_{M,+} = \tilde{C}_{1,-} = \tilde{C}_{M,+} = 0 $, 
and 
\begin{equation*}
    {}^t \! \vec{C} := (C_{2,-}, \cdots, C_{M,-}, C_{1,+}, \cdots, C_{M-1,+}) \in \mathbf{R}^{2M-2}
\end{equation*}
and 
\begin{equation*}
    {}^t \! \vec{\widetilde{C}} := (\widetilde{C}_{2,-}, \cdots, \widetilde{C}_{M,-}, \widetilde{C}_{1,+}, \cdots, \widetilde{C}_{M-1,+}) \in \mathbf{R}^{2M-2}
\end{equation*}
are given as the unique solutions to
the following equations:
\begin{equation}\label{LM1}
    \begin{split}
   \begin{pmatrix}
       -e^{\alpha_-^+ \mathbf{y}}
       + e^{-\alpha_-^- \mathbf{y}} \, {}^t \! J & e^{\alpha_+^- \mathbf{y}}  - e^{\alpha_+^+ \mathbf{y}}   J \\
       (-e^{\alpha_-^+ \mathbf{y}}  +  e^{-\alpha_-^- \mathbf{y}} \,{}^t \! J )\alpha^+_{-}
       & (e^{\alpha_+^- \mathbf{y}}  -  e^{\alpha_+^+ \mathbf{y}}J  ) \alpha^-_+ 
       \end{pmatrix}
        \vec{C} = \vec{\gamma}, 
    \end{split}
\end{equation}
and 
\begin{equation}\label{LMV}
    \begin{split}
       \begin{pmatrix}
       -e^{\alpha_-^+ \mathbf{y}}
       + e^{-\alpha_-^- \mathbf{y}} \,{}^t \!\! J & e^{\alpha_+^- \mathbf{y}}  -   e^{\alpha_+^+ \mathbf{y}} J \\
       (-e^{\alpha_-^+ \mathbf{y}}  +  e^{-\alpha_-^- \mathbf{y}} \, {}^t \!\! J )\alpha^+_{-}
       & (e^{\alpha_+^- \mathbf{y}}  -   e^{\alpha_+^+ \mathbf{y}}J) \alpha^-_+ 
       \end{pmatrix}
        \vec{\widetilde{C}} = \vec{\widetilde{\gamma}}, 
    \end{split}
\end{equation}
where
\begin{equation*}
    J = \begin{pmatrix}
    0 & 1 & 0 &\cdots & 0  \\
    \vdots  & \ddots & \ddots & \ddots & 
    \vdots\\
    0 & \cdots & 0 & 1 & 0 \\
    0 & \cdots & 0& 0 & 1 \\
    0 & \cdots& 0& 0  & 0   
    \end{pmatrix} =
   \begin{pmatrix}
   0_{\mathbf{R}^{M-2}} & I_{\mathbf{R}^{M-2}\otimes \mathbf{R}^{M-2}} \\
   0 &  ^{t} 0_{\mathbf{R}^{M-2}}
   \end{pmatrix}
    \in \mathbf{R}^{M-1} \otimes \mathbf{R}^{M-1}, 
\end{equation*}
\begin{equation*}
    \mathbf{y}:=\mathrm{diag} ( y_1, \cdots, y_{M-1})   \in \mathbf{R}^{M-1}\otimes\mathbf{R}^{M-1},
\end{equation*}
\begin{equation*}
\alpha_{\pm}^{\pm} := \mathrm{diag} ( \alpha_{\frac{3\pm1}{2},\pm}, \cdots,  \alpha_{\frac{(2M-1)\pm 1}{2}, \pm} ) 
    \in \mathbf{R}^{M-1}\otimes\mathbf{R}^{M-1},
\end{equation*}
\begin{equation*}
   {}^t \!\vec{\gamma} := (\gamma_2-\gamma_1, \cdots, \gamma_M -\gamma_{M-1}), 
   \in \mathbf{R}^{M-1}, 
\end{equation*}
and 
\begin{equation*}
   {}^t \!\vec{\widetilde{\gamma}} := (\widetilde{\gamma}_2-\widetilde{\gamma}_1, \cdots, \widetilde{\gamma}_M -\widetilde{\gamma}_{M-1}), 
   \in \mathbf{R}^{M-1}.
\end{equation*}
\end{lem}

\begin{proof}
Using the Feynman-Kac formula, 
we have that $ z_V $ and $ z_1 $
satisfy 
\if1
we get two Differential equations 
\begin{eqnarray} 
&\begin{cases}
  \frac{\partial u_V(t,y)}{\partial t}
   = \frac{\frac{1}{2}a^2 \partial^2 u_V(t,y)}{ \partial y^2}
        + b \frac{ \partial u_V(t,y)}{\partial y} -(V(y)+D(y,p))u_V(t,y)   \\
  u(0,y) = V(y) 
 \end{cases}\label{netuV}
\\
&\begin{cases}
  \frac{\partial u_1(t,y)}{\partial t}
   = \frac{\frac{1}{2}a^2 \partial^2 u_1(t,y)}{ \partial y^2}
        + b \frac{ \partial u_1(t,y)}{\partial y} -(V(y)+D(y,p))u_1(t,y)   \\
  u(0,y) = 1
 \end{cases}\label{netu1}
\end{eqnarray}
\fi
\begin{align}
-V(y) + rz_V(y) &= \frac{1}{2} a^2 z_V''(y) +bz_V'(y)-(V(y)+D(y,p))z_V(y) \label{ODE1}\\
-1 + rz_1(y) &= \frac{1}{2} a^2 z_1''(y) +bz_1'(y)-(V(y)+D(y,p))z_1(y) \label{ODE2}
\end{align}
for $ y \in \mathbf{R} \setminus \{y_1, \cdots, y_{M-1}\} $, respectively. 
Since $z_1, z_V \in C^1$ (see e.g. \cite[Theorem 6.4.1]{KS}),
\begin{equation*}
  z_* (y_i+) = z_* (y_i-), \quad z'_* (y_i+) = z'_* (y_i-), \quad(i=1,2,...,M-1)
\end{equation*}
for $ * = V, 1 $. 
Specifically, 
\begin{equation}\label{C1}
    \begin{split}
      &  C_{i,+} e^{\alpha_{i,+} y_i} + C_{i,-} e^{\alpha_{i,-} y_i}
        + \gamma_i \\
     & \hspace{3cm}   =  C_{i+1,+} e^{\alpha_{i+1,+} y_i} + C_{i+1,-} e^{\alpha_{i+1,-} y_i}
        + \gamma_{i+1}
    \end{split}
\end{equation}
and 
\begin{equation}\label{C2}
    \begin{split}
     &   C_{i,+} \alpha_{i,+} e^{\alpha_{i,+} y_i} + C_{i,-} \alpha_{i,-} e^{\alpha_{i,-} y_i} \\
     & \hspace{3cm}   =  C_{i+1,+}\alpha_{i+1,+} e^{\alpha_{i+1,+} y_i} + C_{i+1,-} \alpha_{i+1,-} e^{\alpha_{i+1,-} y_i}
    \end{split}
\end{equation}
for $i=1,2,...,M-1$, 
which is equivalent to 
the equation \eqref{LM1}. 
Similarly the equations derived 
for $ \widetilde{C}_{i,\pm} $, 
which is obtained by replacing $ \gamma_i $
with $ \widetilde{\gamma}_i $,
is equivalent to \eqref{LMV}.
\end{proof}

\subsection{A numerical scheme}
As we remarked already, 
the inversion of the matrix 
\begin{equation*}
\begin{split}
\begin{pmatrix}
       -e^{\alpha_-^+ \mathbf{y}}
       + e^{-\alpha_-^- \mathbf{y}} \, {}^t \! \! J & e^{\alpha_+^- \mathbf{y}}  -   e^{\alpha_+^+ \mathbf{y}} J \\
       (-e^{\alpha_-^+ \mathbf{y}}  +  e^{-\alpha_-^- \mathbf{y}} \,{}^t \! \! J )\alpha^+_{-}
       & (e^{\alpha_+^- \mathbf{y}}  -   e^{\alpha_+^+ \mathbf{y}}J) \alpha^-_+ 
       \end{pmatrix}
    \end{split}
\end{equation*}
might become too heavy if $ M $ is large. 
We then propose a numerical scheme to solve 
the equations \eqref{LM1} and \eqref{LMV} which might work well when 
\begin{equation*}
    \delta := \min_{1 \leq i \leq M-2} (y_{i+1}-y_i) 
\end{equation*}
is very small.

For simplicity, we assume that 
$ y_1, y_2, \cdots, y_{M-1} $ are 
equally spaced:
\begin{equation*}
    \delta = y_{i+1}-y_i, i=1,2, \cdots, M-1, 
\end{equation*}
where we can get an exact expression. First let us observe that 
\begin{equation*}
\begin{split}
&
\begin{pmatrix}
       -e^{\alpha_-^+ \mathbf{y}}
       + e^{-\alpha_-^- \mathbf{y}} \, {}^t \! J & e^{\alpha_+^- \mathbf{y}}  -   e^{\alpha_+^+ \mathbf{y}} J \\
       (-e^{\alpha_-^+ \mathbf{y}}  +  e^{-\alpha_-^- \mathbf{y}} \, {}^t \! J )\alpha^+_{-}
       & (e^{\alpha_+^- \mathbf{y}}  -   e^{\alpha_+^+ \mathbf{y}}J) \alpha^-_+ 
       \end{pmatrix}
       \\
       &=\begin{pmatrix}
       I
       -  {}^t \! J e^{\alpha_-^+\delta } & I  -    J e^{\alpha_+^- \delta} \\
       (I  -  {}^t \! J e^{\alpha_-^+\delta })\alpha^+_{-}
       & I  -    J e^{\alpha_+^- \delta} \alpha^-_+ 
       \end{pmatrix}
       \begin{pmatrix}
         -e^{\alpha_-^+ \mathbf{y}} &  0_{\mathbf{R}^{M-1} \otimes \mathbf{R}^{M-1}} \\
         0_{\mathbf{R}^{M-1} \otimes \mathbf{R}^{M-1}} & e^{\alpha_+^- \mathbf{y}} 
       \end{pmatrix}
        \\
       &  =: L(\delta)  \begin{pmatrix}
         -e^{\alpha_-^+ \mathbf{y}} &  0_{\mathbf{R}^{M-1} \otimes \mathbf{R}^{M-1}} \\
         0_{\mathbf{R}^{M-1} \otimes \mathbf{R}^{M-1}} & e^{\alpha_+^- \mathbf{y}} 
       \end{pmatrix}.
    \end{split}
\end{equation*}
Since we have immediately 
\begin{equation*}
    \begin{pmatrix}
         -e^{\alpha_-^+ \mathbf{y}} &  0_{\mathbf{R}^{M-1} \otimes \mathbf{R}^{M-1}} \\
         0_{\mathbf{R}^{M-1} \otimes \mathbf{R}^{M-1}} & e^{\alpha_+^- \mathbf{y}} 
       \end{pmatrix}^{-1}
       =     \begin{pmatrix}
         -e^{-\alpha_-^+ \mathbf{y}} &  0_{\mathbf{R}^{M-1} \otimes \mathbf{R}^{M-1}} \\
         0_{\mathbf{R}^{M-1} \otimes \mathbf{R}^{M-1}} & e^{-\alpha_+^- \mathbf{y}} 
       \end{pmatrix},
\end{equation*}
we can concentrate on the inversion 
of $ L (\delta) $. 

Noting that 
\begin{equation*}
\begin{split}
L (\delta) & = 
\begin{pmatrix}
        I-{}^t \! J & I-J \\
        (I-{}^t \! J) \alpha^+_{-} &  (I- J)\alpha^-_{+}
        \end{pmatrix} \\
        & \qquad + \sum_{k=1^\infty}
        \frac{\delta^k}{k!}
        \begin{pmatrix}
                {}^t \! J (\alpha_-^+)^k
                & J (\alpha_+^-)^k \\
                {}^t \! J (\alpha_-^+)^{k+1}
                & J (\alpha_k^-)^{k+1}
        \end{pmatrix},
\end{split}
\end{equation*}
we propose 
the following scheme to invert 
the matrix $ L (\delta) $,
where we actually will need the inversion 
of 
\begin{equation*}
    L(0) = \begin{pmatrix}
        I-{}^t \! J & I-J \\
        (I-{}^t \! J) \alpha^+_{-} &  (I- J)\alpha^-_{+}
        \end{pmatrix}.
\end{equation*}
For convenience, 
we denote
\begin{equation*}
    \begin{pmatrix}
                {}^t \! J (\alpha_-^+)^k
                & J (\alpha_+^-)^k \\
                {}^t \! J (\alpha_-^+)^{k+1}
                & J (\alpha_k^-)^{k+1},
        \end{pmatrix}=: L^{(k)} (0). 
\end{equation*}

\begin{thm}\label{Th4.2}
(i) Let $ \mathbf{c} \in \mathbf{R}^{2(M-1)} $ be fixed. 
Define $ \mathbf{x}_k $,
$ k=0, 1, \cdots $ recursively by
\begin{equation*}
    \mathbf{x}_k = 
    \begin{cases}
    L(0)^{-1} \mathbf{c} &  k=0, \\
    - L(0)^{-1} \sum_{j=0}^{k-1}
    \frac{1}{(k-j)!}L^{(k-j)}(0) \mathbf{x}_j
    & k \geq 2.
    \end{cases}
\end{equation*}
Then $ \mathbf{x} := \sum_{k=0}^\infty \delta^k \mathbf{x}_k $
satisfies 
$ L (\delta)\mathbf{x}=\mathbf{c} $.
(ii) In particular, $ \mathbf{x} - \sum_{k=0}^n \mathbf{x}_k = O (\delta^{n+1}) $ for each $ n $. 
(iii) We have explicitly 
\begin{equation*}
    \begin{split}
        L(0)^{-1}
        &= 
        \begin{pmatrix}
   A_1 
        \alpha^-_+ (I-J)^{-1} & A_1 (I-J)^{-1} \\
          A_2 
        \alpha^+_- (I-{}^t \! J)^{-1} 
        & A_2
        (I- {}^t \! J)^{-1} 
        \end{pmatrix} 
    \end{split}
\end{equation*}
where 
\begin{equation*}
    \begin{split}
    & A_1 = \\
      &  \begin{pmatrix}
         (\alpha_{M,-}-\alpha_{1,+})^{-1} (\hat{\alpha}_- - \hat{\alpha}_+)^{-1}(\hat{\alpha}_+ I - \alpha_{M,-} I) 1_{\mathbf{R}^{M-2}} & (\hat{\alpha}_- - \hat{\alpha}_+)^{-1}
        \\
        (\alpha_{1,+} -\alpha_{M,-})^{-1}
        & {}^t 0_{\mathbf{R}^{M-2}}. 
        \end{pmatrix}
    \end{split}
\end{equation*}
and 
\begin{equation*}
    \begin{split}
    & A_2 = \\
      &   \begin{pmatrix}
        {}^t 0_{\mathbf{R}^{M-2}}
        & (\alpha_{M,-}- \alpha_{1,+})^{-1} \\
        (\hat{\alpha}_+ - \hat{\alpha}_-)^{-1}
        & (\alpha_{M,-}-\alpha_{1,+} )^{-1}  (\hat{\alpha}_+ - \hat{\alpha}_-)^{-1}
    (\hat{\alpha}_{-} I - \alpha_{1,+} I) 1_{\mathbf{R}^{M-2}}
        \end{pmatrix},
    \end{split}
\end{equation*}
with 
\begin{equation*}
    \hat{\alpha}_{\pm} 
    := \mathrm{diag} (\alpha_{2,\pm}, \cdots,  \alpha_{M-1,\pm 1} ) 
    \in \mathbf{R}^{M-2}\otimes\mathbf{R}^{M-2}. 
\end{equation*}
\end{thm}

\begin{proof}
(i), (ii) For each $ n \in \mathbf{N} $ define 
\begin{equation*}
L_n := \sum_{k=0}^n \frac{\delta^k}{k!} L^{(k)} (0). 
\end{equation*}
Then, 
\begin{equation*}
    L -L_n = O (\delta^{n+1}).
\end{equation*}
Since 
\begin{equation*}
    L_n \sum_{k=0}^n \mathbf{x}_k = O (\delta^{n+1}),  
\end{equation*}
by a standard argument we have the assertion (ii), and hence (i). 

(iii) Put
\begin{equation*}
\begin{split}
    K &:= (I-J)^{-1}(I-{}^t \! J)\\
    & = \begin{pmatrix}
    {}^t 0_{\mathbf{R}^{M-2}} & 1 \\
    -I_{\mathbf{R}^{M-2}\otimes \mathbf{R}^{M-2}} & 1_{\mathbf{R}^{M-2}}
    \end{pmatrix}.
\end{split}
\end{equation*}
Then we have that 
\begin{equation*}
    \begin{split}
   &  L(0)^{-1} = 
     \begin{pmatrix}
I-{}^t \! J & I-J \\
        (I-{}^t \! J) \alpha_{-}^+ &  (I- J)\alpha^-_{+}
        \end{pmatrix}^{-1} \\
        &= 
        \begin{pmatrix}
        (\alpha^-_+ K - K \alpha^+_-)^{-1}
        \alpha^-_+ (I-J)^{-1} & (K \alpha^+_- - \alpha^-_+ K)^{-1}(I-J)^{-1} \\
          (\alpha^+_- K^{-1} - K^{-1} \alpha^-_+)^{-1}
        \alpha^+_- (I-{}^t \! J)^{-1} 
        & (K^{-1} \alpha^-_+ -\alpha^+_- K^{-1})^{-1}
        (I- {}^t \! J)^{-1} 
        \end{pmatrix}.
    \end{split}
\end{equation*}
We then see that 
\begin{equation*}
\begin{split}
  & \alpha_+ K - K \alpha_-
 = \begin{pmatrix}
 \alpha_{1, +} & {}^t 0_{\mathbf{R}^{M-2}} \\
0_{\mathbf{R}^{M-2}} & \hat{\alpha}_+
I_{\mathbf{R}^{M-2}\otimes \mathbf{R}^{M-2}} 
\end{pmatrix}\begin{pmatrix}{}^t 0_{\mathbf{R}^{M-2}} & 1 \\ -I_{\mathbf{R}^{M-2}\otimes \mathbf{R}^{M-2}} & 1_{\mathbf{R}^{M-2}}
\end{pmatrix}\\
& \qquad - 
\begin{pmatrix}{}^t 0_{\mathbf{R}^{M-2}} & 1 \\
-I_{\mathbf{R}^{M-2}\otimes \mathbf{R}^{M-2}} & 1_{\mathbf{R}^{M-2}}
\end{pmatrix}
\begin{pmatrix}
\hat{\alpha}_- I_{\mathbf{R}^{M-2}\otimes \mathbf{R}^{M-2}} & 0_{\mathbf{R}^{M-2}} \\
{}^t 0_{\mathbf{R}^{M-2}}
& \alpha_{M,-}
\end{pmatrix} \\
& \qquad \qquad = \begin{pmatrix}{}^t 0_{\mathbf{R}^{M-2}} & 
\alpha_{1,+} -\alpha_{M,-} \\
\hat{\alpha}_- - \hat{\alpha}_+
& (\hat{\alpha}_+ I - \alpha_{M,-} I) 1_{\mathbf{R}^{M-2}} 
\end{pmatrix}
= A_1^{-1}
\end{split}
\end{equation*}
and 
\begin{equation*}
    \begin{split}
      &  \alpha^+_- K^{-1} - K^{-1} \alpha^-_+ \\
      &  = \begin{pmatrix}
        (\hat{\alpha}_{-} I - \alpha_{1,+} I) 1_{\mathbf{R}^{M-2}} & 
        \hat{\alpha}_+ - \hat{\alpha}_-
        \\
        \alpha_{M,-}- \alpha_{1,+}
        & {}^t 0_{\mathbf{R}^{M-2}}
        \end{pmatrix} = A_2^{-1}.
    \end{split}
\end{equation*}
\end{proof}

\section{A Model by the 2-Dimensional Squared Bessel Process}\label{Bessel2}

In this section, 
we specifically assume that the personal condition process $ X_t $ is the 2-dimensional squared Bessel Process
\begin{align}
dX_t= 2 \sqrt{X_t}dW_t + 2 dt \quad(a> 0). \label{squared Bessel}
\end{align}
We further assume that the initial condition distribution is approximated by 
the exponential distribution whose mean is $ \frac{1}{\gamma} $; the limit density
$ f $ in Assumption \ref{a1}
is given by 
\begin{equation*}
f(x) = \gamma e^{-\gamma x}
   \quad ( \gamma > 0) .
\end{equation*}
Moreover, we assume that 
the killing rate functions are
$ V(y) $ and $D(y,p)$ 
\begin{equation*}
\begin{split}
 & V(x) = mx+n, \\
 & D(x,p) = \varphi(p)x+\varrho (p) , \\
 & (m>0,\varphi(p)<0,n,\varrho(p) \in \mathbf{R}
              ).
\end{split}
\end{equation*}
We call this the 2-dimensional Squared Bessel model, 2SB model for short. 
The 2SB model, by nature, satisfies 
Assumptions \ref{a0}, \ref{a1}, and 
\ref{a4}. 

\begin{thm}
The virtual average expected return in the 2SB model is explicitly calculated as:
\begin{equation*}
    \begin{split}
     & \mathrm{VAR}_s (p)
     = \frac{Am}{\lambda} \\
    & + \left( \gamma p + \frac{Amc}{\lambda} -\gamma A n\right)  \frac{1}{\gamma \sqrt{2\lambda }+2\lambda} \\
    & \times 
    \bigg(F\left(1, \frac{1+c}{2\sqrt{2\lambda}}-2,\frac{1+c}{2\sqrt{2\lambda}}-1; -\frac{\gamma-\sqrt{2\lambda }}{\gamma+ \sqrt{2\lambda }}\right) - \frac{1}{c+\sqrt{2\lambda}}\bigg),
    \end{split} 
\end{equation*}
where $ F $ is the hypergeometric function, 
\begin{equation*}
    c:= r +n+ \varrho (p),
\end{equation*}
and 
\begin{equation*}
    \lambda := m + \varphi (p).
\end{equation*}
\end{thm}
\begin{proof}
First note that 
\begin{equation*}
    \begin{split}
     & \mathrm{VAR}_s (p)
     = \int_0^\infty e^{-r t} dt 
     \int_0^\infty \gamma e^{-\gamma x} dx \\
     & \times E [(p -A (m X_t+n) )e^{- (m+\varphi(p)) \int_0^t X_s \,ds}] |X_0 = x] e^{-(n + \varrho(p) ) t}.
    \end{split} 
\end{equation*}
Now we see that 
we need to calculate 
\begin{equation*}
    \begin{split}
    \gamma \int_0^\infty e^{-c t} 
     \int_0^\infty e^{-\gamma x} E [(aX_t+b) e^{- \lambda \int_0^t X_s \,ds} |X_0 = x] \,dx \,dt
    \end{split}
\end{equation*}
for 
\begin{equation*}
    \lambda = m+\varphi(p),
\end{equation*}
\begin{equation*}
    c = r + n + \varphi(p),
\end{equation*}
\begin{equation*}
    a= -Am,
\end{equation*}
and 
\begin{equation*}
    b= p-An,
\end{equation*}
which can be further reduced to 
the calculation of 
\begin{equation*}
    \begin{split}
  I_\gamma (c,\lambda):= \int_0^\infty e^{-c t} 
     \int_0^\infty e^{-\gamma x} E [e^{- \lambda \int_0^t X_s \,ds} |X_0 = x] \,dx \,dt
    \end{split}
\end{equation*}
since 
\begin{equation*}
    \begin{split}
  & \gamma \int_0^\infty e^{-c t} 
     \int_0^\infty e^{-\gamma x} E [X_t e^{- \lambda \int_0^t X_s ,ds} |X_0 = x] \,dx \,dt \\
    & = -\frac{\gamma}{\lambda} \int_0^\infty e^{-c t} 
     \int_0^\infty e^{-\gamma x} 
     \partial_t E [e^{- \lambda \int_0^t X_s \,ds} |X_0 = x] \,dx \,dt \\
     &= -\frac{1}{\lambda}
     -\frac{c}{\lambda}\int_0^\infty e^{-c t} 
     \int_0^\infty e^{-\gamma x} 
     E [e^{- \lambda \int_0^t X_s \,ds} |X_0 = x] \,dx \,dt. 
    \end{split}
\end{equation*}
That is, 
\begin{equation*}
    \begin{split}
          & \mathrm{VAR}_s (p)
     = \frac{Am}{m+ \varphi (p)} \\
    & + \left( \gamma p + \frac{Am(r +n+ \varrho (p))}{m + \varphi (p)} -\gamma A n\right) I_\gamma \left( r +n+ \varrho (p), m + \varphi (p)\right). 
    \end{split}
\end{equation*}

It is well-known that 
(see e.g. Ikeda-Watanabe \cite{IW})
that 
\begin{equation*}
\begin{split}
    E [e^{- \lambda \int_0^t X_s \,ds} |X_0 = x] = 
    \frac{e^{-x \sqrt{2\lambda}\tanh \sqrt{2\lambda} t} }{\cosh\sqrt{2\lambda} t},
\end{split}
\end{equation*}
therefore we have 
\begin{equation*}
    \begin{split}
        I &= \int_0^\infty  \frac{e^{-ct}}{(\gamma+ \sqrt{2\lambda}\tanh \sqrt{2\lambda}t)\cosh\sqrt{2\lambda}t}\,dt \\
        &= \int_0^\infty  \frac{e^{-ct}}{(\gamma \cosh \sqrt{2\lambda}t + \sqrt{2\lambda}\sinh \sqrt{2\lambda}t)}\,dt \\
        &= \int_0^\infty  \frac{2}{\gamma+ \sqrt{2\lambda }} \frac{e^{-(c+\sqrt{2\lambda })t}}{1+ \frac{\gamma-\sqrt{2\lambda }}{\gamma+ \sqrt{2\lambda }}e^{-2\sqrt{2\lambda}t}}\,dt \\
        & =  \int_0^\infty  \frac{2e^{-(c+\sqrt{2\lambda })t}}{\gamma+ \sqrt{2\lambda }} \sum_{j=1}^\infty \ \left(-\frac{\gamma-\sqrt{2\lambda }}{\gamma+ \sqrt{2\lambda }}e^{-2t}\right)^j\,dt \\
     & =  \frac{2}{\gamma+ \sqrt{2\lambda }} \sum_{j=1}^\infty\left(-\frac{\gamma-\sqrt{2\lambda }}{\gamma+ \sqrt{2\lambda }}\right)^j \int_0^\infty e^{-(c+ (2j+1) \sqrt{2\lambda})t} \,dt \\
     & =\frac{2}{\gamma+ \sqrt{2\lambda }} \sum_{j=1}^\infty\left(-\frac{\gamma-\sqrt{2\lambda }}{\gamma+ \sqrt{2\lambda }}\right)^j \frac{1}{c+ (2j+1) \sqrt{2\lambda}}\\
     &= \frac{1}{\gamma \sqrt{2\lambda }+ 2\lambda } \sum_{j=1}^\infty\left(-\frac{\gamma-\sqrt{2\lambda }}{\gamma+ \sqrt{2\lambda }}\right)^j 
     \frac{\left(\frac{1+c}{2\sqrt{2\lambda}}-2 \right)_j}{\left(\frac{1+c}{2\sqrt{2\lambda}}-1 \right)_j} \\
     &=  \frac{1}{\gamma \sqrt{2\lambda }+ 2\lambda } \times \\
     & \bigg(F\left(1, \frac{1+c}{2\sqrt{2\lambda}}-2,\frac{1+c}{2\sqrt{2\lambda}}-1; -\frac{\gamma-\sqrt{2\lambda }}{\gamma+ \sqrt{2\lambda }}\right) - \frac{1}{c+\sqrt{2\lambda}}\bigg). 
    \end{split}
\end{equation*}

Note that here $ (\cdot)_n$ is the Pochhammer symbol,
that is, 
\begin{equation*}
    (x)_n = \prod_{k=0}^{n-1} (x+ k) 
\end{equation*}
for a complex number $ x $. 
\end{proof}

\section{Concluding Remark}

The present paper proposed a totally new framework to evaluate the heterogeneous risks in whole-life insurance. 
We have employed 
a large-agent limit which is analogous to the thermodynamic one, 
and both the life-time and the surrender-time 
are modelled by the killing time of the diffusion process, 
while the cash-flow is evaluated 
by the Laplace transform. 
The two specific models have shown 
the potential of our framework. 

Even though we have worked only on the determination of the level-premium, 
the proposed framework can be used for more general cases, 
including forward-looking models
where a mean-field type approximation can work.

\appendix \section{Appendix} 
 \label{continuous death}

\begin{lem}\label{appendlem}
We have that 
\begin{equation}\label{appendlemf}
\mathbb{E}[e^{-r\zeta^i}1_{\{\zeta^i\leq t\}}]
=\int_0^t   e^{-rs}\mathbb{E}[V(X^i_s)e^{-\int^s_0V(X^i_u)du}]ds.
\end{equation}
\end{lem}

\begin{proof}
Let $G(u)=\mathbb{P}(\zeta^i > u)$ for $ u \geq 0 $.
Then, by taking the expectation of both sides of \eqref{a1f},
\begin{equation*}
    G(s) = \mathbb{E} [e^{-\int_0^s V(X_u) \,du }], \quad s \geq 0. 
\end{equation*}
Since $ e^{-\int_0^s V(X_u) \,du } $
is differentiable in $ s $ almost surely
and uniformly bounded by $ 1 $, 
we see that  $ G $ is also differentiable and
\begin{equation*}
    G'(s) = - \mathbb{E}[V(X^i_s)e^{-\int^s_0V(X^i_u)du}], \quad s > 0.
\end{equation*}
Then, since
\begin{align*}
\mathbb{E}\left [e^{-r\zeta^i}1_{\{\zeta^i \leq t\}} \right]&= \int_0^t  e^{-rs}\mathbb{P}(\zeta^i \in   ds)\\
&= - \int_0^t  e^{-rs} d G_s = - \int_0^t  e^{-rs} G' (s) \,ds, 
\end{align*}
we get \eqref{appendlemf}. 
\end{proof}
\end{document}